\documentclass[11pt,a4]{article}

\usepackage{enumerate}
\usepackage{latexsym} 
\usepackage{theorem}
\usepackage{graphics}
\usepackage{graphicx}
\usepackage{amsmath}
\usepackage{amssymb}

\newtheorem{defeng}{Definition}[section]
\newtheorem{theorem}[defeng]{Theorem}

\newtheorem{lemma}[defeng]{Lemma}

\newtheorem{corollary}[defeng]{Corollary}
{\theorembodyfont{\rmfamily} }
{\theorembodyfont{\rmfamily} }
{\theorembodyfont{\rmfamily} }
{\theoremstyle{break}\theorembodyfont{\rmfamily} }
{\theoremstyle{break}\theorembodyfont{\rmfamily} }

\newcounter{claim}

\newenvironment{proof}[1][]%
 {\noindent {\setcounter{claim}{0}\sc proof #1 --- }{}}{\hfill$\Box$\vspace{2ex}} 

{\refstepcounter{claim}\vspace{1ex}\noindent{(\it\arabic{claim}){#1}{}}\it}{\vspace{1ex}}

	{\noindent {}{#1}{}}{ This proves~(\arabic{claim}).\vspace{1ex}}

\usepackage{ifpdf}

\ifpdf
\fi

\begin{document}

\title{Vertex elimination orderings for hereditary graph classes}

\author{Pierre Aboulker\thanks{Universit\'e Paris 7 -- Paris Diderot,
    LIAFA,\ email: pierre.aboulker@liafa.jussieu.fr,
    pierre.charbit@liafa.jussieu.fr} , Pierre Charbit\footnotemark[1] ,\\
  Nicolas Trotignon\thanks{CNRS, LIP, ENS Lyon, UCBL, Universit\'e de
    Lyon, INRIA\ email: nicolas.trotignon@ens-lyon.fr. Partially
    supported by the French Agence Nationale de la Recherche under
    reference \textsc{anr-14-blan-STINT.}}~~and Kristina
  Vu\v{s}kovi\'c\thanks{School of Computing, University of Leeds,
    Leeds LS2 9JT, UK, and Faculty of Computer Science (RAF), Union
    University, Knez Mihajlova 6/VI, 11000 Belgrade, Serbia. \ email:
    k.vuskovic@leeds.ac.uk.  Partially supported by EPSRC grant
    EP/K016423/1 and Serbian Ministry of Education and Science
    projects 174033 and III44006.  \newline The authors are also
    supported by PHC Pavle Savi\'c grant, jointly awarded by EGIDE, an
    agency of the French Minist\`ere des Affaires \'etrang\`eres et
    europ\'eennes, and Serbian Ministry of Education and Science.
    \newline The first, second and fourth authors are partially
    supported by the French Agence Nationale de la Recherche under
    reference \textsc{anr-10-jcjc-Heredia.}\newline This work was
    supported by the LABEX MILYON (ANR-10-LABX-0070) of Universit\'e
    de Lyon, within the program "Investissements d'Avenir"
    (ANR-11-IDEX-0007) operated by the French National Research Agency
    (ANR).}}

\maketitle


\begin{abstract}
  We provide a general method to prove the existence and compute
  efficiently elimination orderings in graphs.  Our method relies on
  several tools that were known before, but that were not put together
  so far: the algorithm LexBFS due to Rose, Tarjan and Lueker, one of
  its properties discovered by Berry and Bordat, and a local
  decomposition property of graphs discovered by Maffray, Trotignon
  and Vu\v skovi\'c. 
\end{abstract}

\section{Introduction}\label{sec:intro}

In this paper all graphs are finite and simple. A graph $G$ {\em
  contains} a graph $F$ if $F$ is isomorphic to an induced subgraph of
$G$.  A class of graphs is \emph{hereditary} if for every graph $G$ of
the class, all induced subgraphs of $G$ belong to the class.  A graph
$G$ is {\em $F$-free} if it does not contain $F$. When ${\cal F}$ is a
set of graphs, $G$ is {\em ${\cal F}$-free} if it is $F$-free for
every $F \in {\cal F}$.  Clearly every hereditary class
of graphs is equal to the class of ${\cal F}$-free graphs for some
${\cal F}$ (${\cal F}$ can be chosen to be the set of all graphs not in
the class but all induced subgraphs of which are in the class). The
induced subgraph relation is not a well quasi order (contrary for
example to the minor relation), so the set $\cal F$ does not need to
be finite.

When $X\subseteq V(G)$, we write $G[X]$ for the subgraph of $G$
induced by $X$.  An ordering $(v_1, \dots, v_n)$ of the vertices of a
graph $G$ is an \emph{$\cal F$-elimination ordering} if for every
$i=1, \dots, n$, $N_{G[\{ v_1, \dots, v_i\} ]}(v_i)$ is $\cal F$-free.
Note that this is equivalent to the existence, in every induced
subgraph of $G$, of a vertex whose neighbourhood is $\cal F$-free.

Let us illustrate our terminology on a classical example.  We denote
by $S_2$ the independent graph on two vertices.  A vertex is {\em
  simplicial} if its neighborhood is $S_2$-free, or equivalently
induces a clique.  A graph is \emph{chordal} if it is hole-free, where
a {\em hole} is a chordless cycle of length at least $4$.

\begin{theorem}[Dirac \cite{dirac:chordal}]
  \label{th:d}
  Every chordal graph admits an $\{S_2\}$-elimination ordering.
\end{theorem}

\begin{theorem}[Rose, Tarjan and Lueker \cite{rose.tarjan.lueker:lbfo}]
  \label{th:t}
  There exists a linear-time algorithm that computes an
  $\{S_2\}$-elimination ordering of an input chordal graph.
\end{theorem}

\subsection*{Motivation, goals, and outline of the paper}

We believe that elimination orderings are important, because several
classical hereditary classes, such as perfect graphs or even-hole-free
graphs, admit deep decomposition theorems that are hard to use for
algorithmic purposes.  For more details, we send the reader to surveys
(\cite{nicolas:perfect} for perfect graphs and
\cite{evenholefreegraphs} for even-hole-free graphs). 
Sometimes, as we shall see, 
the existence of a vertex with some local structural property is
more useful for design of efficient algorithms than a global description
of the class.  Even for chordal graphs that are rather well
structured, elimination orderings are the basis for the fastest
algorithms. 

Our goal here is to give a general method to prove the existence of
elimination orderings, to compute them efficiently and to use them
to design algorithms solving problems for different hereditary classes of graphs.
Our method relies on two main ingredients that are not new but that
were not put together before:

\begin{enumerate}
\item LexBFS, a classical algorithm discovered by Rose, Tarjan and
  Lueker~\cite{rose.tarjan.lueker:lbfo}, and some of its properties
  discovered by Berry and Bordat~\cite{berry.b:lexBFS:dirac}. 
\item A local decomposition property of graphs discovered by Maffray, Trotignon and
  Vu\v skovi\'c~\cite{maffray.t.v:3pcsquare}.   This property is
  called \emph{Property ($\star$)} in \cite{maffray.t.v:3pcsquare}, but here we give it a
  more meaningful name of \emph{local decomposability}. 
\end{enumerate}

In Section~\ref{sec:lexBFS}, we explain the first ingredient, and in
Section~\ref{sec:sa} the second.  We conclude Section~\ref{sec:sa} by
illustrating how our method reproduces the classical proofs of
Theorems~\ref{th:d} and \ref{th:t}, so that we may consider the rest
of our work as a generalization of these.

In Section~\ref{sec:nsa} we give two classes of graphs for which the
existence of an ${\cal F}$-elimination ordering is proved in previous
works (namely even-hole-free graphs and square-theta-free Berge
graphs).  We explain for each of them how our method can be used prove
the existence of the order.  For even-hole-free graphs, our method
leads to speeding up the algorithm that computes a maximum clique.  To
be more specific, it turns out that the classes in
Section~\ref{sec:nsa} are slight generalizations of even-hole-free
graphs and square-theta-free Berge graphs, defined by excluding different 
Truemper configurations, that are special types of graphs 
(defined formally at the end of this section)
that play an
important role in the study of hereditary graph classes 
(see 
survey~\cite{vuskovic:truemper}).  This fact is interesting to us,
especially because Truemper configurations appear also in the following
section.

In Section~\ref{sec:tc}, we apply systematically our method to produce
classes of graphs that admit $\cal F$-elimination orderings for all
possible non-empty sets of graphs $\cal F$ made of non-complete graphs
on three vertices (there are seven such sets $\cal F$).  This leads us
to define seven classes of graphs, each of which having its own
elimination ordering by our method.  Two of these classes were
previously studied (namely universally signable graphs and wheel-free
graphs) and five of them are new.  For almost all these classes, we
get something from the ordering: a bound on the chromatic number, a
coloring algorithm, or an algorithm for the maximum clique problem.
To our great surprise, this systematic application of the method outlined
in this paper leads
again to classes that are all defined by excluding some Truemper configurations.

Section~\ref{sec:oq} is devoted to open questions. 

We now sum up the previously known optimization algorithms for which
we get better complexity (each time, we improve the
previously known complexity by at least a factor of $n$):

\begin{itemize}
\item Maximum weighted clique in even-hole-free graphs in time $O(nm)$.
\item Maximum weighted clique in universally signable graphs in time $O(n+m)$. 
\item Coloring in universally signable graphs in time $O(n+m)$. 
\end{itemize}

\subsection*{Terminology and notation}

For $x\in V(G)$, $N(x)$ denotes the set of neighbors of $x$, and
$N[x]=N(x) \cup \{x\}$.  For a set of vertices $S$, $N(S)$ denotes the
set of vertices not in $S$ that have a neighbor in $S$, and $N[S] = S
\cup N(S)$.  For $S \subseteq V(G)$, $G[S]$ denotes the subgraph of
$G$ induced by $S$, and $G - S=G[V(G)-S]$.

Recall that a {\em hole} in a graph is a chordless cycle of length at least 4,
where the {\em length} of a hole is the number of its edges.
A hole is {\em even} or {\em odd} according to the parity of its length.

Sometimes, we consider \emph{weighted graphs}, which are graphs given
with a non-negative weight for every vertex. The weight of a subset of
vertices is then the sum of the weights of its elements. The usual
problem of finding a maximum clique generalizes to weighted graphs to
the problem of finding a clique of maximum weight.

In all complexity analysis of the algorithms, $n$ denotes the number
of vertices of the input graph, and $m$ the number of edges. We say
that an algorithm runs in \emph{linear time} if its complexity is
$O(n+m)$.

\subsection*{Truemper configurations}

Special types of graphs that are called Truemper configurations appear
in different sections of this work, so let us define them now.  A {\em
  3-path configuration} is a graph induced by three internally vertex
disjoint paths of length at least~1, $P_1=x_1\ldots y_1$,
$P_2=x_2\ldots y_2$ and $P_3=x_3\ldots y_3$, such that either $x_1 =
x_2 = x_3$ or $x_1, x_2, x_3$ are all distinct and pairwise adjacent,
and either $y_1 = y_2 = y_3$ or $y_1, y_2, y_3$ are all distinct and
pairwise adjacent.  Furthermore, the vertices of $P_i\cup P_j$, $i\neq
j$, induce a hole.  Note that this last condition in the definition
implies the following.

\begin{itemize}
\item If $x_1,x_2,x_3$ are distinct (and therefore pairwise adjacent)
  and $y_1,y_2,y_3$ are distinct, then the three paths have length at
  least 1.  In this case, the configuration is called a \emph{prism}.
\item If $x_1 = x_2 = x_3$ and $y_1 = y_2 = y_3$, then the three paths
  have length at least 2 (since a path of length~1 would form a chord of the
  cycle formed by the two other paths).  In this case, the
  configuration is called a \emph{theta}.
\item If $x_1 = x_2 = x_3$ and $y_1,y_2,y_3$ are distinct, or if
  $x_1,x_2,x_3$ are distinct and $y_1 = y_2 = y_3$, then at most one
  of the three paths has length 1, and the others have length at
  least~2.  In this case, the configuration is called a
  \emph{pyramid}.
\end{itemize}

A \emph{wheel} $(H,v)$ is a graph formed by a hole $H$, called the
\emph{rim}, and a vertex $v$, called the \emph{center}, such that the
center has at least three neighbors on the rim.  A \emph{Truemper
  configuration} is a graph that is either a prism, a theta, a pyramid
or a wheel (see Figure~\ref{f:tc}).

\begin{figure}
\begin{center}\includegraphics[height=2cm]{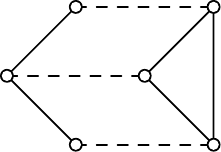}
    \hspace{.2em}
    \includegraphics[height=2cm]{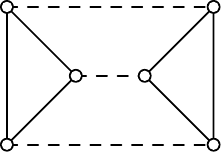}
    \hspace{.2em}
    \includegraphics[height=2cm]{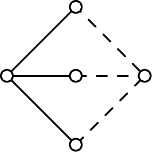}
    \hspace{.2em}
    \includegraphics[height=2cm]{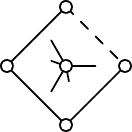}
\end{center}
  \caption{Pyramid, prism, theta and wheel (dashed lines represent paths)\label{f:tc}}
\end{figure}

\section{A theorem on LexBFS orderings}
\label{sec:lexBFS}

LexBFS is a linear time algorithm of Rose, Tarjan and
Lueker~\cite{rose.tarjan.lueker:lbfo} whose input is any graph $G$
together with a vertex $s$, and whose output is a linear ordering of
the vertices of $G$ starting at $s$.  A linear ordering of the
vertices of a graph $G$ is a \emph{LexBFS ordering} if there exists a
vertex $s$ of $G$ such that the ordering can be produced by LexBFS
when the input is $G, s$.  As the reader will soon see, we do not need
to define LexBFS more precisely.  The purpose of this section is to
provide an alternative proof of the following result.

\begin{theorem}[Berry and Bordat \cite{berry.b:lexBFS:dirac}]
  \label{lastlexbfs}
  If $G$ is a non-complete graph and $z$ is the last vertex of
  a LexBFS ordering of $G$, then there exists a connected component
  $C$ of $G - N[z]$ such that for every neighbor $x$ of $z$,
  either $N[x]=N[z]$, or $N(x) \cap C \neq \emptyset$.
\end{theorem}

Equivalently, if we put $z$ together with its neighbors of the first
type, the resultant set of vertices is a clique, a homogeneous set,
and its neighborhood is a minimal separator. Such sets are called
\emph{moplexes} in~\cite{berry.b:lexBFS:dirac} and
Theorem~\ref{lastlexbfs} is stated in term of moplexes
in~\cite{berry.b:lexBFS:dirac}.  We find it more convenient for our
purpose to state it as we do.  We now give an alternative proof of
Theorem~\ref{lastlexbfs} for several reasons.  First, it is
shorter than the original proof mainly because it relies on the 
following nice characterization of LexBFS orderings instead of the
full description of the algorithm. Second, we believe that
Lemma~\ref{regle1} that we use in our proof and that was not stated
explicitly before is of independent interest.

\begin{theorem}[Brandst{\"a}dt, Dragan and Nicolai~\cite{brandstadtDN:97}]
  \label{regle0}
  An ordering $\prec$ of the vertices of a graph $G=(V,E)$ is a LexBFS
  ordering if and only if it satisfies the following property: for all
  $a,b,c\in V$ such that $c \prec b \prec a$, $ca \in E$ and $cb
  \notin E$, there exists a vertex $d$ in $G$ such that $d\prec c$,
  $db \in E$ and $da \not\in E$.
\end{theorem}

Let us strengthen a little this property for our purposes.

\begin{lemma}\label{regle1}
  Let $\prec$ be a LexBFS ordering of  a graph $G=(V,E)$. 
  Let $z$ denote the last vertex in this
  ordering.  Then for all vertices $a,b,c\in V$ such that $c
  \prec b \prec a$ and $ca \in E$, there exists a path from $b$ to
  $c$ whose internal vertices are disjoint from $N[z]$.
\end{lemma}

\begin{proof}
  By contradiction assume there exists such a triple $c\prec b \prec
  a$ for which no such path exists from $b$ to $c$. Choose this triple
  to be minimal with respect to the sum of the positions of its
  elements in the ordering. Observe that since $b$ cannot be adjacent
  to $c$, by Theorem \ref{regle0} there is a vertex $d$ such that $d
  \prec c$, $db \in E$ and $da \not\in E$. There must be a path $P$
  from $c$ to $d$ whose internal vertices are disjoint from $N[z]$
  otherwise $d\prec c\prec b$ would contradict the minimality of
  $c\prec b\prec a$. Since $db\in E$, $d$ must be a neighbor of $z$
  otherwise $P \cup \{d\}$ is a path that contradicts the
  hypothesis. In particular, $z \neq a$.  So we can apply Theorem
  \ref{regle0} to the triple $d\prec a \prec z$. Thus there is a
  vertex $e$ such that $e\prec d$, $ea \in E$ and $ez \notin E$. But
  again by minimality of $c\prec b\prec a$, there exist two paths, one
  from $e$ to $c$ (from the triple $e\prec c \prec a$), and one from
  $e$ to $b$ (from the triple $e\prec b \prec a$) whose internal
  vertices are disjoint from $N[z]$.  Since $e$ is a non-neighbor of
  $z$, the union of these paths contains a path from $b$ to $c$ whose
  internal vertices are disjoint from $N[z]$, a contradiction.
\end{proof}

With this lemma, we are now able to easily prove the aforementioned
theorem.

\begin{proof}[of Theorem \ref{lastlexbfs}] We denote by $\prec$ the
  LexBFS ordering. First observe that $V(G)-N[z]\neq \emptyset$, since
  otherwise by Theorem \ref{regle0} $G$ is complete.
  
  Let $x$ be a neighbor of $z$, and assume that
  $N(x)\subseteq N[z]$. To show  that in this case $N[x]=N[z]$, let $y$ be
  another neighbour of $z$, and assume $xy\notin E(G)$. Either $x\prec
  y$ or $y\prec x$, but in both cases, Lemma \ref{regle1} with $a=z$
  implies the existence of a neighbor of $x$ that is not a neighbor
  of $z$, contradicting the assumption.

  Now assume that $N(x)\not\subseteq N[z]$. 
  Denote by $u$ the last vertex in $\prec$ that does not belong
  to $N[z]$, and by $C$ the connected component of $G - N[z]$
  containing $u$. We now show that $C$ is the desired component. If
  $x\prec u$, then by Lemma \ref{regle1} applied to 
  $x\prec u\prec z$,
  there exists a path from $x$ to $u$ which does not meet $N[z]$, so
  $x$ must have a neighbor in $C$. 
  So we may assume that $u\prec x$ and that $u$ is not adjacent to $x$.
  Since $x$ has a
  neighbor $u'$ not belonging to $N[z]$, we must have $u'\prec u$.
  Now by Lemma \ref{regle1} applied to $u'\prec u\prec x$, 
  $u$ and $u'$ belong to the same component $C$.  
  \end{proof}

\section{Locally $\cal F$-decomposable graphs}
\label{sec:sa}

Let $\cal F$ be a set of graphs. We are interested in  graphs $G$
that admit $\cal F$-elimination orders (which is equivalent to say
that every induced subgraph of $G$ has a vertex whose neighborhood is
$\cal F$-free).  A much stronger property is the one of being {\em
  locally ${\cal F}$-free} : every vertex of $G$ has a $\cal F$-free
neighborhood. The following property, that sits between these two,  
was introduced by Maffray, Trotignon and Vu\v{s}kovi\'c
in~\cite{maffray.t.v:3pcsquare} (where it was called Property~$(\star)$).

\begin{defeng}
  Let $\cal F$ be a set of graphs. A graph $G$ is \emph{locally $\cal
  F$-decomposable} if for every vertex $v$ of $G$, every $F\in \cal F$
  contained in $N(v)$ and every connected component $C$ of
  $G-N[v]$, there exists $y\in F$ such that $y$ has a non-neighbor in
  $F$ and no neighbors in $C$. 
  
  A class of graphs ${\cal C}$ is {\em locally ${\cal
      F}$-decomposable} if every graph $G\in {\cal C}$ is locally
  ${\cal F}$-decomposable.
\end{defeng}

It is easy to see that if a graph is locally $\cal F$-decomposable,
then so are all its induced subgraphs.  Therefore, for all sets of
graphs $\cal F$, the class of graphs that are locally ${\cal
  F}$-decomposable is hereditary.

Observe that a complete graph is locally $\cal F$-decomposable for any
set of graphs $\cal F$.  On the other hand, a complete graph may fail
to have an $\cal F$-elimination ordering, but this happens only when
$\cal F$ contains graphs that are complete.  This is why in
the next theorem and in all the applications to come, we require that
no graph of $\cal F$ is complete.  

Here is now our main result.  A similar theorem was given
in~\cite{maffray.t.v:3pcsquare} with another kind of ordering (not worth defining
here) instead of LexBFS.  This ordering was also lexicographic in some
sense, but it could not be computed in linear time.

\begin{theorem}
  \label{th:fldp-elim}
  If ${\cal F}$ is a set of non-complete graphs, and $G$ is a locally
  ${\cal F}$-decomposable graph, then every LexBFS ordering of $G$ is
  an $\cal F$-elimination ordering.
\end{theorem}

\begin{proof}
  Let $z$ be the last vertex of a LexBFS ordering of $G$. If $G$ is
  complete, then $N(z)$ is $\cal F$-free because no graph of $\cal F$
  is complete.  Otherwise, the connected component $C$ given by
  Theorem~\ref{lastlexbfs} is such that
  every vertex of $N(z)$ that has
  non-neighbors in $N(z)$ has a neighbor in $C$.  So by definition of local $\cal
  F$-decomposability, $N(z)$ must be $\cal F$-free.

  Therefore, any LexBFS ordering is an $\cal F$-elimmination ordering, because if
  $(v_1,v_2,\ldots,v_n)$ is a LexBFS ordering, then for all $i$,
  $(v_1,v_2,\ldots,v_i)$ is a LexBFS ordering of
  $G[\{v_1,v_2,\ldots,v_i\}]$ (this follows for instance from the
  characterization of LexBFS orderings given in Theorem~\ref{regle0}).
\end{proof}

Let us now illustrate how Theorem~\ref{th:fldp-elim} can be used with
the simplest possible set made of non-complete graphs: ${\cal
  F}=\{S_{2}\}$, where $S_{2}$ is the independent graphs on two
vertices.  The following is  of course well known, but we
write its proof to illustrate our notions.

\begin{lemma}
  \label{s2ldp}
  A graph $G$ is locally $\{S_2\}$-decomposable if and only if $G$ is
  chordal.
\end{lemma}

\begin{proof}
  Suppose $G$ is not locally $S_2$-decomposable.  Then
  for some $x\in V(G)$ and some connected component $C$ of $G-N[x]$,
  $G[N(x)]$ contains an induced subgraph $F$ isomorphic to $S_2$, and
  every vertex of $F$ has a neighbour in $C$. This clearly implies
  that $G$ contains a hole.

  To prove the converse, suppose that $G$ contains a hole $H$, and let
  $y,x,z$ be three consecutive vertices of $H$.  Let $C$ be the
  connected component of $G-N[x]$ that contains the vertices of $H-\{
  x,y,z\}$. Then $\{ y,z\}$ is an $S_2$ of $N(x)$, and both $y$ and
  $z$ have neighbors in $C$. Therefore, $G$ is not
  locally $S_2$-decomposable.
\end{proof}

Hence, as promised in the introduction, we reobtain the proof for
Theorems~\ref{th:d} and~\ref{th:t} by using  Lemma~\ref{s2ldp} and
Theorem~\ref{th:fldp-elim}.

\section{Even-hole-free graphs and perfect graphs}
\label{sec:nsa}

In this section, we show how local decomposability can be used to
provide elimination orderings and algorithms for even-hole-free graphs
and some Berge graphs.  A graph $G$ is {\em Berge} if $G$ and
$\overline{G}$ are odd-hole-free.  In the last few decades much
research was devoted to the study of Berge graphs, odd-hole-free
graphs and even-hole-free graphs (for surveys
see~\cite{nicolas:perfect,evenholefreegraphs}).  For all these classes
global decomposition theorems are known.  Most famously the celebrated
proof of the Strong Perfect Graph Conjecture (which states that {\em a
  graph is perfect if and only if it is Berge}) obtained in 2002 by
Chudnovsky, Robertson, Seymour and Thomas~\cite{chudnovsky.r.s.t:spgt}
is based on a decomposition theorem for Berge graphs.  Also
decomposition theorems were obtained for even-hole-free
graphs~\cite{conforti.c.k.v:eh1}, the most precise one by da Silva and
Vu\v skovi\'c~\cite{dasv}.  Unfortunately, up to now, no one knows how
these decomposition theorems can be used to design fast algorithms for
optimization problems.

The results that we present here are in fact proved for
generalizations of Berge graphs and even-hole-free graphs, the
so-called {\em signed graphs}.  We want to state their definitions
here, because we find it interesting that they make use of the same kind
of obstructions as the classes of graphs in the next section.  A graph
is {\em odd-signable} if there exists an assignment of $0,1$ weights
to its edges that makes every chordless cycle of odd weight.  A graph
is {\em even-signable} if there exists an assignment of $0,1$ weights
to its edges that makes every triangle of odd weight and every
hole of even weight.  In~\cite{truemper} Truemper proved a
theorem that characterizes graphs whose edges can be assigned $0,1$
weights so that chordless cycles have prescribed parities. The
characterization states that this can be done for a graph $G$ if and
only if it can be done for all Truemper configurations and $K_4$'s contained in
$G$. An easy consequence of this theorem when applied to odd-signable
and even-signable graphs gives the following characterizations of
these classes (see \cite{cckv-cap}).  A \emph{sector} of a wheel is a
subpath of the rim of length at least~1 whose ends are adjacent to the
center and whose internal vertices are not.  A wheel is \emph{even} if
it has an even number of sectors, and it is \emph{odd} if it has an
odd number of sectors of length 1.

\begin{itemize}
\item A graph is \emph{odd-signable} if and only if it is (theta, prism,
  even wheel)-free.
\item A graph is \emph{even-signable} if and only if it is (pyramid,
  odd wheel)-free.
\end{itemize}

We are now ready to obtain two results on vertex elimination orderings
by using local $\cal F$-decomposability.  These results were known already
(see  \cite{dsv:triangulatedvertexevenholefree} and \cite{maffray.t.v:3pcsquare}), 
and were obtained by a special kind of lexicographic ordering
of the vertices that is different from LexBFS (but more closely related to decomposition).  
Proving the existence
of the ordering directly from Theorem~\ref{th:fldp-elim} allows in both
cases for the desired  ordering to be computed in
linear-time.   A \emph{4-hole} is a hole of length 4. 

\begin{theorem} [da Silva and Vu\v{s}kovi\'c
  \cite{dsv:triangulatedvertexevenholefree}]
  \label{triangulated1}
  4-hole-free odd-signable graphs are locally hole-decomposable.
\end{theorem}

Theorems~\ref{triangulated1} and~\ref{th:fldp-elim} directly imply
that 4-hole-free odd-signable graphs admit a hole-elimination
ordering.  Theorem~\ref{triangulated1} is used
in~\cite{dsv:triangulatedvertexevenholefree} to obtain a robust
$O(n^2m)$-time algorithm for computing a maximum weighted clique in a
4-hole-free odd-signable graph (and hence in an even-hole-free graph).
We now show how to reduce this complexity to ${O}(nm)$.

\begin{theorem}\label{algo}
  There is an $O(nm)$-time algorithm whose input is a weighted graph
  $G$ and whose output is a maximum weighted clique of $G$ or a
  certificate proving that $G$ is not 4-hole-free odd-signable.
\end{theorem}

\begin{proof}
  Let ${\cal H}$ denotes the class of all holes and consider the
  following algorithm.  Compute in linear time a LexBFS ordering
  $(v_1, \ldots ,v_n)$ of $G$.  By Theorems~\ref{th:fldp-elim} and
  \ref{triangulated1}, this ordering is an ${\cal H}$-elimination
  ordering if $G$ is a 4-hole-free odd-signable graph.  Testing
  whether a graph is chordal can be done in linear time
  \cite{rose.tarjan.lueker:lbfo}, and hence it can be checked in
  ${O}(nm)$-time whether $(v_1, \ldots ,v_n)$ is an ${\cal
    H}$-elimination ordering.

  So, we may assume that $(v_1, \dots, v_n)$ is an ${\cal
    H}$-elimination ordering of $G$.  We suppose inductively that a
  maximum weighted clique of $G[\{ v_1, \dots, v_{n-1}\} ]$ is found
  in time $O((n-1)m)$.  A maximum weighted clique of $G[N[v_n]]$ can
  be found in time $O(m)$ \cite{rose.tarjan.lueker:lbfo}.  So, we now
  know a maximum weighted clique of $G[N[v_n]]$ and a maximum weighted
  clique of $G[\{ v_1, \dots, v_{n-1}\} ]$.  A maximum weighted clique
  among these is a maximum weighted clique of $G$.  All this takes
  time $O((n-1)m) + O(m) = O(nm)$.
\end{proof}

We now turn our attention to Berge graphs (or more precisely to
even-signable graphs that generalize them).
A {\em square-theta} is a theta that contains a 4-hole.  A \emph{long
  hole} is a hole of length at least~5.

\begin{theorem}[Maffray, Trotignon and Vu\v{s}kovi\'c
  \cite{maffray.t.v:3pcsquare}]\label{t2}
  Square-theta-free even-signable graphs are locally
  long-hole-decomposable.
\end{theorem}

Again, Theorems~\ref{t2} and~\ref{th:fldp-elim} directly imply that
square-theta-free even-signable graphs admit a long-hole-elimination ordering.
Based on Theorem~\ref{t2} an ${O}(n^7)$-time robust algorithm is given
in~\cite{maffray.t.v:3pcsquare} for computing a maximum weighted
clique in a square-theta-free Berge graph (note that this class
generalizes both 4-hole-free Berge graphs and claw-free Berge graphs).
It relies on a long-hole-elimination ordering.  With the machinery
presented here, we can obtain this ordering in linear time, but
unfortunately, this does not improve the overall complexity of the
maximum clique algorithm.

\section{Seven generalizations of chordal graphs}
\label{sec:tc}

In this section we apply systematically our method to all possible
sets made of non-complete graphs of order~3.  This leads to 
seven classes of graphs, two of which were studied before (namely
universally signable graphs and wheel-free graphs). 

To describe the classes of graphs that we obtain, we need to be more
specific about wheels.  A wheel is a \emph{1-wheel} if for some
consecutive vertices $x, y, z$ of the rim, the center is adjacent to
$y$ and non-adjacent to $x$ and $z$.  A wheel is a \emph{2-wheel} if
for some consecutive vertices $x, y, z$ of the rim, the center is
adjacent to $x$ and $y$, and non-adjacent to $z$.  A wheel is a
\emph{3-wheel} if for some consecutive vertices $x, y, z$ of the rim,
the center is adjacent to $x$, $y$ and $z$.  Observe that a wheel can
be simultuneously a 1-wheel, a 2-wheel and a 3-wheel.  On the other hand, every wheel
is a 1-wheel, a 2-wheel or a 3-wheel.  Also, any 3-wheel is either a
2-wheel or a \emph{universal wheel} (that is a wheel whose center is
adjacent to all vertices of the rim).

Up to isomorphism, there are four graphs on three vertices, and three
of them are not complete.  These three graphs (namely the independent
graph on three vertices denoted by $S_3$, the path of length~2 denoted
by $P_3$ and its complement denoted by $\overline{P_3}$) are studied
in the next lemma.

\begin{lemma}\label{l:main}
  For a graph $G$ the following hold.
  \begin{itemize}
  \item[(i)] $G$ is locally $\{S_3\}$-decomposable if and
    only if $G$ is (1-wheel, theta, pyramid)-free.
  \item[(ii)] $G$ is locally  $\{P_3\}$-decomposable if and
    only if $G$ is 3-wheel-free.
  \item[(iii)] $G$ is locally $\{\overline{P_3}\}$-decomposable if and
    only if $G$ is (2-wheel, prism, pyramid)-free.
  \end{itemize}
\end{lemma}

\begin{proof}
  To prove (i), first observe that if $G$ contains a 1-wheel, theta or
  pyramid $H$, then $H$ contains vertices $v,x,y,z$ such that $\{
  x,y,z\}$ induces an $S_3$, $\{ x,y,z \}\subseteq N(v)$, and $H'=H-\{
  v,x,y,z\}$ is a connected subgraph of $G-N[v]$ such that every
  vertex of $\{ x,y,z\}$ has a neighbour in $H'$.

  To prove the converse, let $v\in V(G)$ be such that $G[N(v)]$
  contains $S_3$, and $C$ a component of $G-N[v]$ such that every
  vertex of $S_3$ has a neighbor in $C$.  Denote by $x, y, z$ the
  three members of $S_3$.  Let $P$ be a chordless path from $x$ to $y$
  with interior in $C$.  Let $Q$ be a chordless path from $z$ to $z'$,
  such that $V(Q)-\{z\} \subseteq C$, $z'$ has neighbors in the
  interior of $P$, and is of minimum length among such paths
  (possibly, $Q=z=z'$).

  Suppose that at least one of $x$ or $y$ has neighbors in $Q$ (this
  implies that $Q$ has length at least~1).  Call $w$ the vertex of $Q$
  closest to $z$ along $Q$, that has neighbors in $\{x, y\}$, and
  suppose up to symmetry that $w$ is adjacent to $y$.  Call $w'$ the
  vertex of $Q$ closest to $z$ along $Q$ that has neighbors in
  $P-y$. Call $x'$ the neighbor of $w'$ in $P$, closest to $x$ along
  $P$.  Now, $V(xPx') \cup V(zQw') \cup \{v, y\}$ induces a theta or a
  1-wheel centered at $y$.

  Therefore, we may assume that none of $x, y$ has a neighbor in
  $Q$. If $z'$ has a unique neighbor in $P$, then $V(P) \cup V(Q) \cup
  \{v\}$ induces a theta.  If $z'$ has exactly two neighbors in $P$
  that are adjacent, then $V(P) \cup V(Q) \cup \{v\}$ induces a
  pyramid.  Otherwise, $V(P) \cup V(Q) \cup \{v\}$ contains a theta.
  
  \vspace{2ex}
  
  To prove (ii), first observe that if $G$ contains a 3-wheel $H$,
  then $H$ contains vertices $v,x,y,z$ such that $x,y,z$ is a $P_3$,
  $\{ x,y,z \}\subseteq N(v)$, and $H'=H-\{ v,x,y,z\}$ is a connected
  subgraph of $G-N[v]$ such that both $x$ and $z$ have a neighbor in
  $H'$.  To prove the converse, let $v\in V(G)$ be such that $G[N(v)]$
  contains a chordless path $xyz$, and $C$ a component of $G-N[v]$
  such that $x$ and $z$ both have a neighbor in $C$.  Then clearly $C
  \cup \{ v,x,y,z\}$ contains a 3-wheel.

  \vspace{2ex}

  To prove (iii), first observe that if $G$ contains a 2-wheel, prism
  or pyramid $H$, then $H$ contains vertices $v,x,y,z$ such that $\{
  x,y,z\}$ induces a $\overline{P_3}$, $\{ x,y,z \}\subseteq N(v)$,
  and $H'=H-\{ v,x,y,z\}$ is a connected subgraph of $G-N[v]$ such
  that every vertex of $\{ x,y,z\}$ has a neighbor in $H'$.

  To prove the converse, let $v\in V(G)$ be such that $G[N(v)]$
  contains $\overline{P_3}$, and $C$ a component of $G-N[v]$ such that
  every vertex of $\overline{P_3}$ has a neighbor in $C$.  Denote by
  $x, y, z$ the vertices of $\overline{P_3}$ in such a way that $xy$
  is the only edge of $G[\{ x, y, z\} ]$.  Let $P$ be a path from $x$ to $y$
  with interior in $C$ whose unique chord is $xy$.  Let $Q$ be a
  chordless path from $z$ to $z'$, such that $V(Q)-\{z\} \subseteq C$,
  $z'$ has neighbors in the interior $P$, and is of minimum length
  among such paths (possibly, $Q=z=z'$).

  Suppose that at least one of  $x$ or $y$ has neighbors in $Q$.
  Call $w$ the vertex of $Q$ closest to $z$ along $Q$, that has
  neighbors in $\{x, y\}$, and suppose up to symmetry that $w$ is
  adjacent to $y$.  Call $w'$ the vertex of $Q$ closest to $z$ along
  $Q$ that has neighbors in $P-y$. Call $x'$ the neighbor of $w'$ in
  $P$, closest to $x$ along $P$.  Now, $V(xPx') \cup V(zQw') \cup \{v,
  y\}$ induces a 2-wheel centered at $y$. 

  Therefore, we may assume that none of $x, y$ has a neighbor in $Q$.
  If $z'$ has a unique neighbor in $P$, then $V(P) \cup V(Q) \cup
  \{v\}$ induces a pyramid or a 2-wheel (when $P$ has length~2).  If
  $z'$ has exactly two neighbors in $P$ that are adjacent, then $V(P)
  \cup V(Q) \cup \{v\}$ induces a prism.  Otherwise, $V(P) \cup V(Q)
  \cup \{v\}$ contains a pyramid.
\end{proof}

The next lemma allows us to combine the results of the previous one. 

\begin{lemma}
  \label{l:union}
  Let ${\cal F}, {\cal F}', {\cal H}, {\cal H'}$ be sets of graphs
  such that $\cal F$ and $\cal F'$ contain only non-complete graphs.
  Suppose that the class of locally ${\cal F}$-decomposable graphs is
   equal to the class of ${\cal H}$-free graphs, and that the class
  of locally ${\cal F}'$-decomposable graphs is equal to the class of
  ${\cal H}'$-free graphs.  Then, the class of locally $({\cal F} \cup
  {\cal F}')$-decomposable graph is is equal to the class of $({\cal
    H} \cup {\cal H}')$-free graphs.
\end{lemma}

\begin{proof}
  Suppose that $G$ is locally $({\cal F} \cup {\cal
    F}')$-decomposable.  From the definition of local decomposability,
  it follows that $G$ is locally ${\cal F}$-decomposable and locally
  ${\cal F}'$-decomposable.  Hence, $G$ is both ${\cal H}$-free and
  ${\cal H}'$-free.  It is therefore $({\cal H} \cup {\cal H}')$-free.

  Suppose conversely that $G$ is $({\cal H} \cup {\cal H}')$-free.
  Then $G$ is ${\cal H}$-free and ${\cal H}'$-free.  It is thereofore
  locally ${\cal F}$-decomposable and locally ${\cal
    F}'$-decomposable.  From the definition of local decomposability,
  it follows that $G$ is locally $({\cal F} \cup {\cal
    F}')$-decomposable.
\end{proof}

\begin{table}\begin{center}
\begin{tabular}{cccc}
  $i$&Class ${\cal C}_i$&${\cal F}_i$&Neighborhood\\\hline
  1 &\parbox[c]{4cm}{
  \begin{center} (1-wheel, theta, 
      pyramid)-free
      \end{center}
      }&$\left\{\parbox[c]{.55cm}{\includegraphics[width=.5cm]{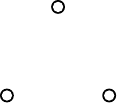}}\right\}$
  &\parbox[c]{4cm}{\begin{center}no stable set of size 3\end{center}}\\
  2 & 3-wheel-free&$\left\{\parbox[c]{.55cm}{\includegraphics[width=.5cm]{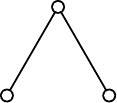}}\right\}$
  &\parbox[c]{4cm}{\begin{center}disjoint union of cliques\end{center}}\\
  3 &\parbox[c]{4cm}{\begin{center} (2-wheel,  prism, 
      pyramid)-free\end{center}}&$\left\{\parbox[c]{.55cm}{\includegraphics[width=.5cm]{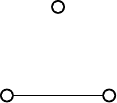}}\right\}$
  &\parbox[c]{4cm}{\begin{center}complete multipartite\end{center}}\\
  4 &\parbox[c]{4cm}{\begin{center} (1-wheel, 3-wheel, theta, 
      pyramid)-free\end{center}}&$\left\{\parbox[c]{1.3cm}{\includegraphics[width=.5cm]{fig-7.pdf}\,, \includegraphics[width=.5cm]{fig-8.pdf}}\right\}$
  &\parbox[c]{4cm}{\begin{center}disjoint union of at most two cliques\end{center}}\\
  5 &\parbox[c]{4cm}{\begin{center} (1-wheel, 2-wheel,  prism,
      theta, pyramid)-free\end{center}}&$\left\{\parbox[c]{1.3cm}{\includegraphics[width=.5cm]{fig-7.pdf}\,, \includegraphics[width=.5cm]{fig-9.pdf}}\right\}$
  &\parbox[c]{4cm}{\begin{center}stable sets of size at most 2 with
      all possible edges between them\end{center}}\\
  6 &\parbox[c]{4cm}{\begin{center} (2-wheel, 3-wheel, prism, 
      pyramid)-free\end{center}}&$\left\{\parbox[c]{1.3cm}{\includegraphics[width=.5cm]{fig-8.pdf}\,, \includegraphics[width=.5cm]{fig-9.pdf}}\right\}$
  &\parbox[c]{4cm}{\begin{center}clique or stable set\end{center}}\\
  7 &\parbox[c]{4cm}{\begin{center} (wheel, prism, theta, pyramid)-free\end{center}}&$\left\{\parbox[c]{2.2cm}{\includegraphics[width=.5cm]{fig-7.pdf}\,, \includegraphics[width=.5cm]{fig-8.pdf}\,, \includegraphics[width=.5cm]{fig-9.pdf}}\right\}$
  &\parbox[c]{4cm}{\begin{center}clique or stable set of size~2\end{center}}\\  8 & hole-free&
  $\left\{\parbox[c]{.3cm}{\ \includegraphics[height=.5cm]{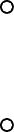}\
    }\right\}$  &\parbox[c]{4cm}{\begin{center}clique\end{center}}\\
\end{tabular}
\end{center}
\caption{Eight classes of graphs\label{t:c}}
\end{table}

\begin{figure}
  \center
  \includegraphics{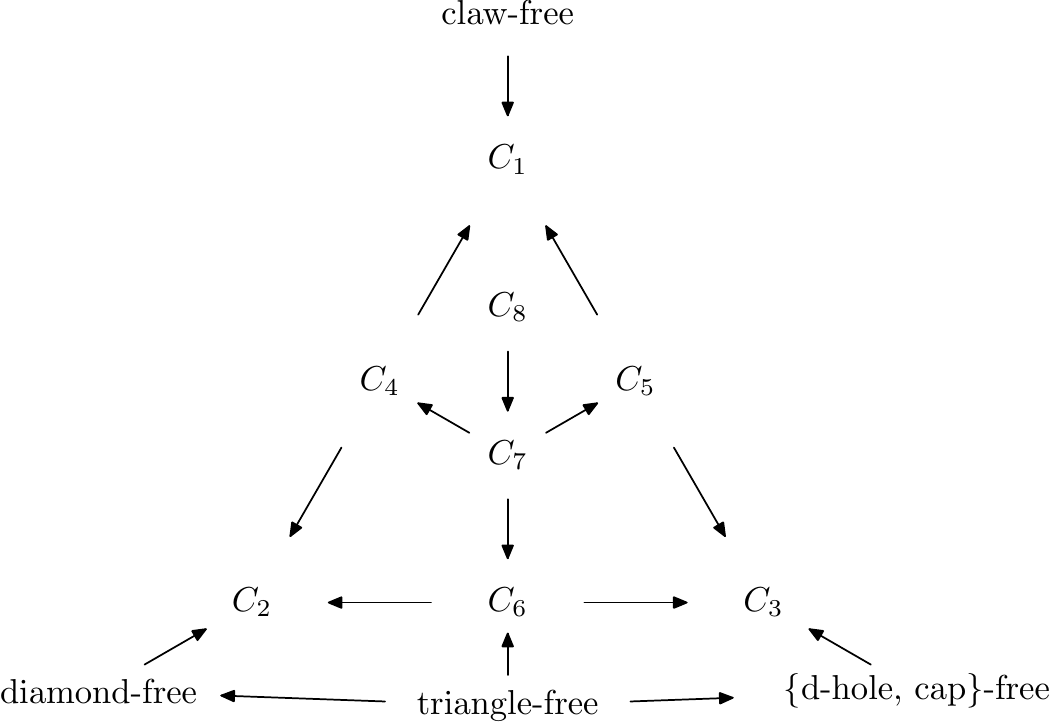}
  \caption{Inclusion for several classes of graphs. An arrow from $A$ to
    $B$ means ``$A$ is contained in $B$''.  Arrows arising from
    transitivity are not represented.\label{f:inc}}
\end{figure}

Table~\ref{t:c} describes eight different classes of graphs ${\cal
  C}_1, \dots, {\cal C}_8$, all defined by excluding induced subgraphs
described in the second column of the table (one of them is the class
of chordal graphs that we put back to have a complete picture).  The
third column describes a class ${\cal F}_i$.  The last column
describes the class of ${\cal F}_i$-free graphs.  Inclusions between
our classes and several known classes are represented in
Figure~\ref{f:inc} (where the \emph{diamond} is the graph obtained
from $K_4$ by removing one edge, a \emph{cap} is cycle of length at
least~5 with a unique chord joining two vertices at distance 2 on the
cycle, a \emph{d-hole} is 3-wheel such that the center has degree~3,
and the \emph{claw} is $K_{1, 3}$).  Observe that a d-hole is also a
2-wheel. The following theorem follows directly from
Lemmas~\ref{l:main}, \ref{l:union} and~\ref{s2ldp}.

\begin{theorem}
  For $i=1, \dots, 8$, let ${\cal C}_i$ and ${\cal F}_i$ be the
  classes defined as in Table~\ref{t:c}.  For $i= 1, \dots, 8$, the
  class ${\cal C}_i$ is exactly the class of locally ${\cal
    F}_i$-decomposable graphs.
\end{theorem}

With Theorem \ref{th:fldp-elim}, this directly implies the following.

\begin{theorem}
  \label{th:m}
  For $i=1, \dots, 8$, let ${\cal C}_i$ and ${\cal F}_i$ be the
  classes defined as in Table~\ref{t:c}. Then every LexBFS ordering of
  a graph of ${\cal C}_i$ is an ${\cal F}_i$-elimination ordering.
\end{theorem}

\subsection*{Known classes}

We now describe the two classes of graphs from Table~\ref{t:c}
that (apart from chordal graphs) were studied before.  The first one is
${\cal C}_7$, i.e. graphs that contain no Truemper configuration, or
equivalently by Theorem~\ref{th:m}, graphs that are ${\cal
  F}_{7}$-locally decomposable. These are studied in~\cite{cckv-u},
where they are called \emph{universally signable graphs}.  The
existence of a vertex whose neighborhood is ${\cal F}_7$-free given
by Theorem~\ref{th:m} is exactly the following theorem
from~\cite{cckv-u}, that was originally proved through a global
decomposition theorem.  Theorem~\ref{th:m} provides a shorter proof as
well as an algorithm that outputs the order 
that does not rely on global decomposition.  In the next subsection,
we study several algorithmic consequences.

\begin{theorem}[Conforti, Cornu\'ejols, Kapoor and Vu\v{s}kovi\'c  \cite{cckv-u}]
  \label{cckv-u-1}
  Every non-empty universally signable graph contains a simplicial
  vertex or a vertex of degree $2$.
\end{theorem}

The second class that was studied previously is the class of
wheel-free graphs and its super-class ${\cal C}_2$.  These might have
interesting structural properties, as suggested by several subclasses,
see~\cite{aboulkerRTV:propeller} for example for a list of them.  The
next theorem (which follows from Theorem~\ref{th:m} for $i=2$) states
the only non-trivial property that is known to be satisfied by all
wheel-free graphs.  The original proof (due to Chudnovsky who
communicated it to us but did not publish it) is by induction, and the
proof relying on our method is much shorter.

\begin{theorem}[Chudnovsky \cite{c:perso}]
  \label{th:maria}
  Every non-empty 3-wheel-free graph contains a vertex whose
  neighborhood is a disjoint union of cliques.
\end{theorem}

The following corollary extends a well-known fact: a chordal graph $G$
has at most $n$ maximal cliques.

\begin{corollary}
  \label{col:m}
  A 3-wheel-free graph $G$ has at most $m$ maximal cliques. 
\end{corollary}

\begin{proof}
  Induction on $m$.  By Theorem~\ref{th:maria}, consider a vertex $v$
  of degree $d$ whose neighborhood is a disjoint union of cliques.  By
  the induction hypothesis, $G-v$ has at most $m-d$ maximal cliques, and
  because of its neighborhhood, $v$ is in at most $d$ maximal cliques.
\end{proof}

\subsection*{Consequences}\label{sec:sa1}

\begin{table}\begin{center}
\begin{tabular}{ccccc}
  $i$&$\chi$-bounded&Max clique&Coloring\\\hline
  \rule{0cm}{.6cm}1&$f(x) = O(x^2/\log x)$&NP-hard~\cite{poljak:74}&NP-hard~\cite{holyer:81}\\
  \rule{0cm}{.6cm}2&No~\cite{Zyk49}&$O(nm)$~\cite{rose.tarjan.lueker:lbfo}&NP-hard~\cite{maffray.preiss:triangle}\\
  \rule{0cm}{.6cm}3&No~\cite{Zyk49}&$O(nm)$&NP-hard~\cite{maffray.preiss:triangle}\\
  \rule{0cm}{.6cm}4&$f(x) = 2x-1$&$O(n+m)$&?\\
  \rule{0cm}{.6cm}5&$f(x) = 2x-1$&$O(nm)$&?\\
  \rule{0cm}{.6cm}6&No~\cite{Zyk49}&$O(n+m)$&NP-hard~\cite{maffray.preiss:triangle}\\
  \rule{0cm}{.6cm}7&$f(x) = \max(3, x)$ \cite{cckv-u}&$O(n+m)$&$O(n+m)$\\
  \rule{0cm}{.6cm}8&$f(x)=x$ \cite{dirac:chordal}&$O(n+m)$ \cite{rose.tarjan.lueker:lbfo}&$O(n+m)$ \cite{rose.tarjan.lueker:lbfo}\\
\end{tabular}
\end{center}
\caption{Several properties of classes defined in Table~\ref{t:c}\label{t:p}}
\end{table}

Table~\ref{t:p} describes several properties of the classes defined in
Table~\ref{t:c}.  We indicate a reference for the properties that are
already known, or follow easily from the given references.  
Let us now explain and prove all these properties.

\vspace{2ex}

Let us analyze the column ``$\chi$-bounded'' of Table~\ref{t:p}.  A
hereditary class of graphs is
\emph{$\chi$-bounded}~(see~\cite{gyarfas:perfect}) if for some
function $f$, every graph $G$ in the class satisfies $\chi(G) \leq
f(\omega(G))$. The column indicates whether the class $C_i$ is
$\chi$-bounded, and if so, gives the smallest known function proving
so.  Classes ${\cal C}_2, {\cal C}_3$ and ${\cal C}_6$ are not
$\chi$-bounded because they contain all triangle-free graphs, and
these may have arbitrarily large chromatic number as first shown by
Zykov~\cite{Zyk49}.  For classes ${\cal C}_1, {\cal C}_4$ and ${\cal
  C}_5$, we may rely on degeneracy.  Say that a hereditary class of
graphs is \emph{$\omega$-degenerate} if there exists a function $g$
such that every non-empty graph in the class has a vertex of degree at
most $g(\omega(G))$.  It is easy to check that by the greedy coloring
algorithm, if a hereditary class of graphs is $\omega$-degenerate with
a non-decreasing function $g$, then it is $\chi$-bounded with function
$g+1$.  The function given for classes ${\cal C}_4$ and ${\cal C}_5$
follows from the fact that these classes are clearly
$\omega$-degenerate with function $g(x) = 2x-2$.  For the class ${\cal
  C}_1$, we use Ramsey theory.  Kim~\cite{kim:95} proved that for some
constant $c$, every graph on $ct^2/\log t$ vertices admits a stable
set of size~3 or a clique of size $t$.  Therefore, the vertex whose
neighborhood is $S_3$-free in any graph in ${\cal C}_1$ proves that
${\cal C}_1$ is $\omega$-degenerate with function $g(x) = O(x^2/\log
x)$.  Observe that the results in this paragraph just improve bounds.
Indeed, a theorem due to K\"uhn and Osthus \cite{kuhnOsthus:04} proves
that theta-free graphs (and therefore graphs in ${\cal C}_1$, ${\cal
  C}_4$ and ${\cal C}_5$) are $\omega$-degenerate, but their function
is quite big.

\vspace{2ex}

Let us now analyze the column ``Max clique'' of Table~\ref{t:p}, that
gives the best complexity of finding a maximum weighted clique in a
graph of the corresponding class.  By a result of
Poljak~\cite{poljak:74}, it is NP-hard to compute a maximum stable set
in a triangle-free graph.  Rephrased in the complement, it is NP-hard
to compute a maximum clique in an $S_3$-free graph, and therefore in
graphs from ${\cal C}_1$.  Finding a maximum weighted clique in ${\cal
  C}_2$ is easy as follows: for every vertex $v$, look for a maximum
weighted clique in $N(v)$, and choose the best clique among these.
This can be implemented by running $n$ times the $O(n+m)$ algorithm of
Rose, Tarjan and Lueker, because $N(v)$ is chordal for every $v$.  In
fact, this algorithm works in the larger class of universal-wheel-free
graphs.

For ${\cal C}_4$, we need to be careful about the complexity
analysis. Here is an algorithm that finds a maximum (weighted) clique
in $G \in {\cal C}_4$.  First by Theorem~\ref{th:m}, we find in linear
time an $\{S_3, P_3\}$-elimination ordering of $G$, say $(v_1, \dots,
v_n)$.  This means that in $G[\{v_1, \dots, v_i\}]$, $N(v_i)$ is a
disjoint union of at most two cliques.  We now show that, having this
order, we can compute a maximum clique in time $O(m)$.  We may assume
that $G$ is connected (otherwise we work on components separately), so
$m\geq n-1$.  Suppose inductively that a maximum clique of $G[\{ v_1,
\dots, v_{n-1}\} ]$ is found in time $O(m-d(v_n))$.  We now take the
vertices of $N(v_n)$ one by one.  We give name $x$ and label $X$ to
first one, and check whether the next ones are adjacent to $x$.  If
so, we give them label $X$.  If some are not adjacent to $x$, we give
name $y$ and label $Y$ to the first one that we meet.  The next
vertices receive label $X$ or $Y$ according to their adjacency to $x$
or $y$.  Note that exactly one of these adjacencies must occur, since
$N(v_n)$ is the union of at most two cliques.  At the end of this
loops, the vertices with label $X$ and $Y$ form  at most two
cliques in $N(v_n)$.  They are identified in time $O(d(v_n))$.  So, we
now know all the maximal cliques of $G[N[v_n]]$ and a maximum clique
of $G[\{ v_1, \dots, v_{n-1}\} ]$.  A maximum clique among these is a
maximum clique of $G$.  All this takes time $O(m-d(v_n)) + O(d(v_n)) =
O(m)$.  Observe that this algorithm relies on a constant time checking
of the adjacency, so it needs the graph to be represented by an
adjacency matrix.  Therefore, the time complexity is $O(n+m)$, but the
space complexity is $O(n^2)$.  Observe also that this algorithm is not
robust. If the input graph is not in ${\cal C}_4$, the output is a set
of vertices, and if it is a clique, we cannot be sure that it has
maximum weight.  Since ${\cal C}_7$ is a subclass of ${\cal C}_4$, we
obtain an algorithm for the maximum clique problem for universally
signable graphs that is fastest than the $O(nm)$-time algorithm that
follows from~\cite{cckv-u}.

For class ${\cal C}_6$, the algorithm is similar to the previous one.
We have to find a maximum clique in $N(v_n)$ in time $O(d(v_n))$.  It
is easy to verify quickly whether the neighbohood of $v_n$ is a
clique or a stable set, and in both cases, it is immediate to find in
time $O(d(v_n))$ a maximum weighted clique in it. We omit further
details.

For ${\cal C}_3$ (that contains ${\cal C}_5$), the algorithm is
similar to the previous one, except that we rely on a
$\{\overline{P_3}\}$-elimination ordering of $G$ instead of an $\{S_3,
P_3\}$-elimination ordering.  As a result, the neighborhood of the
last vertex $v$ is complete multipartite.  We do not know how to find
a maximum clique in $N(v)$ in time $O(d(v))$, so we do not know how to
obtain a linear time algorithm.  Instead, we look for a maximum clique
in $N(v)$ in time $O(m)$, and therefore the overall complexity is
$O(nm)$.

\vspace{2ex}
 
Let us now analyze the column ``Coloring'' of Table~\ref{t:p}, that
gives the best complexity for coloring a graph of the corresponding
class.  Since the edge-coloring problem is NP-hard~\cite{holyer:81},
it follows that coloring line graphs is NP-hard, and therefore,
so is coloring claw-free graphs (that are all in ${\cal C}_1$).  Classes
${\cal C}_2, {\cal C}_3$ and ${\cal C}_6$ contain all triangle-free
graphs, that are NP-hard to color as proved by Preissmann and
Maffray~\cite{maffray.preiss:triangle}.  For ${\cal C}_7$, we first
try to find a 2-coloring of the graph by the classical BFS algorithm.
If it does not exist, we look for a $\max(3, \omega(G))$-coloring of
the input graph $G$ as follows. By Theorem~\ref{th:m} we obtain an
$\{S_3, P_3, \overline{P_3}\}$-elimination ordering in linear time.
As a result, the neighborhood of the last vertex of the ordering is a
clique or has size 2.  We remove the last vertex $v$, color recursively
the remaining vertices, and give some available color to $v$.

\section{Open questions}
\label{sec:oq}

Addario-Berry, Chudnovsky, Havet, Reed and
Seymour~\cite{addarioBerryEtAl:ehf} proved that every even-hole-free
graph admits a vertex whose neighborhood is the union of two cliques.
We wonder whether this result can be proved by some search algorithm. 

Corollary~\ref{col:m} suggests that a linear time algorithm for the
maximum clique problem might exists in ${\cal C}_2$, but we could not
find it.

We are not aware of a polynomial time coloring algorithm for graphs in
${\cal C}_4$ or ${\cal C}_5$, but it would be surprising to us that it
exists.

Since class ${\cal C}_1$ generalizes claw-free graphs, it is natural
to ask which of the properties of claw-free graphs it has, such as a
structural description (see \cite{chudnovsky.seymour:claw4}), a
polynomial time algorithm for the maximum stable set
(see~\cite{faenzaOrioloStauffer:clawFree}), approximation algorithms
for the chromatic number~(see \cite{king:these}), a polynomial time
algorithm for the induced linkage problem (see~\cite{fialaKLP:12}),
and a polynomial $\chi$-binding function (see~\cite{gyarfas:perfect}).
Also we wonder whether theta-free graphs are $\chi$-bounded by a
\emph{polynomial} (quadratic?) function (recall that in~\cite{kuhnOsthus:04},
they are proved to be $\chi$-bounded).

In~\cite{cckv-u}, an $O(nm)$ time algorithm is described for the
maximum weighted stable set problem in ${\cal C}_7$.  Since the class
is a simple generalization of chordal graphs, we wonder whether a
linear time algorithm exists.

\section*{Ackowledgement}

Thanks to Maria Chudnovsky for indicating to us Theorem~\ref{th:maria}
and its proof,  which was the starting point of this research.  Thanks
to Michael Rao and two anonymous referees for comments that helped improve this 
paper.

\end{document}